\documentclass[a4paper, 11pt]{article}

\usepackage[linkcolor=black,colorlinks=true,citecolor=black,filecolor=black]{hyperref}

\usepackage{amsmath,amssymb,amsthm}

\usepackage[a4paper, margin=1.2in]{geometry}

\newtheorem{theorem}{Theorem}

\newtheorem{lemma}[theorem]{Lemma}

\newtheorem{definition}[theorem]{Definition}
\newtheorem{observation}[theorem]{Observation}

\usepackage{textcomp}
\usepackage{graphicx,graphics}

\newcommand{\paragraphprev}{}
\newcommand{\paragraphdef}[1]{\paragraph{#1}}

\newcommand{\introductionAppendixExplanation}{
In the main part, we show a bound $O(1.321^{n})$ that still improves on the bound of~\cite{istt10}. In the appendix we prove the better bound. The only change is we use a better result of~\cite{istt10} which has different parameters; however these are not not stated explicitly, so we need to derive and prove them.
}
\newcommand{\theoremAppendixExplanation}{
The previous theorem is proved in the appendix.
We prove the following weaker theorem in the main section:
}

\newcommand{\valueGammaH}{\ensuremath{0.6073995502}}
\newcommand{\valueOneMinusGammaH}{\ensuremath{0.3926004498}}
\newcommand{\valueBetaH}{\ensuremath{0.8022563838}}
\newcommand{\valueOneHalfToTheBetaH}{\ensuremath{0.57345159}}
\newcommand{\forced}{\ensuremath{{\rm forced}}}
\newcommand{\guessed}{\ensuremath{{\rm guessed}}}
\newcommand{\vbl}{\ensuremath{{\rm vbl}}}

\newcommand{\E}{\ensuremath{\mathbf{E}}}

\newcommand{\poly}{\ensuremath{{\textup{\rm poly}}}}

\newcommand{\ignore}[1]{}

\usepackage{cancel}
\usepackage{algorithm}
\usepackage{algorithmic}

\usepackage{wrapfig}

\newcommand{\sat}{\mathrm{sat}}
\newcommand{\PPZ}{\textsc{PPZ}}
\newcommand{\PPSZ}{\textsc{PPSZ}}
\newcommand{\ol}{\overline}

\newcommand{\Resolve}{\textsc{Resolve}}

\newcommand{\Schoening}{\textsc{Schoening}}
\newcommand{\Comb}{\textsc{Comb}}

\newcommand{\ISTT}{\textsc{ISTT}}
\newcommand{\ISTTSch}{\textsc{ISTTSch}}

\newcommand{\floor}[1]{\left\lfloor{}#1\right\rfloor}

\newcommand{\refHtheory}{Sections 4.2 and 4.3}
\newcommand{\refHevaluation}{Section 4.6}

\title{Improving PPSZ for $3$-SAT using Critical Variables}

\author{
Timon Hertli\\
\texttt{timon.hertli@inf.ethz.ch}\\
\and
Robin A. Moser\\
\texttt{robin.moser@inf.ethz.ch}\\
\and
Dominik Scheder\\
\texttt{dominik.scheder@inf.ethz.ch}\\
\vspace*{3mm}\\
Institute for Theoretical Computer Science\\
Department of Computer Science\\
ETH Z\"urich, 8092 Z\"urich, Switzerland
}

\begin{document}

\maketitle

\begin{abstract}
  A critical variable of a satisfiable CNF formula is a variable that has the
  same value in all satisfying assignments.  Using a simple case
  distinction on the fraction of critical variables of a CNF formula, we
  improve the running time for
  3-SAT from $\mathcal{O}\!\left(1.32216^{n}\right)$ by Rolf~\cite{rolf2006}
  to $\mathcal{O}\!\left(1.32153^{n}\right)$.
  Using a different approach, Iwama et al.~\cite{istt10} very recently achieved a running time of $\mathcal{O}\!\left(1.32113^{n}\right)$.
  Our method nicely combines with theirs, yielding the currently fastest known algorithm with running time
  $\mathcal{O}\!\left(1.32065^{n}\right)$.
  We also improve the
  bound for 4-SAT from
  $\mathcal{O}\!\left(1.47390^{n}\right)$~\cite{it04} to
  $\mathcal{O}\!\left(1.46928^{n}\right)$, where
  $\mathcal{O}\!\left(1.46981^{n}\right)$ can be obtained using the
  methods of~\cite{it04} and~\cite{rolf2006}.
\end{abstract}

\section{Introduction}

The ideas behind the most successful algorithms for $k$-SAT are
surprisingly simple.  In 1999, Paturi, Pudl\'ak, and Zane~\cite{ppz}
proposed the following algorithm.  Given a $k$-CNF formula $F$, we
choose a variable $x$ uniformly at random from the $n$ variables in
$F$, choose a truth value $b \in \{0,1\}$, and set $x$ to $b$, thereby
replacing $F$ by $F^{[x \mapsto b]}$, and continue with $F^{[x \mapsto
  b]}$.  The value $b$ is chosen as follows: If the formula contains
the unit clause $(x)$, we choose $b=1$. If it contains $(\bar{x})$, we
choose $b=0$.  In these two cases, we say $x$ was {\em forced}. If it
contains neither, we choose $b$ randomly and say $x$ was {\em
  guessed}.  Finally, if the formula contains both $(x)$ and
$(\bar{x})$, we can give up, since the formula is unsatisfiable. This
algorithm is usually called $\PPZ$ after its three inventors.

Intuitively, if $F$ is ``strongly constrained'', then the algorithm
encounters many unit clauses, hence it needs to guess significantly
fewer than $n$ variables. On the other hand, if  $F$ is only
``weakly constrained'', it has multiple satisfying assignments,
making it easier to find one.
Paturi, Pudl\'ak and Zane~\cite{ppz} make this intuition precise and
show that $\PPZ$ finds a satisfying assignment for a $k$-CNF formula with probability at
least $2^{-(1-1/k)n}$, provided there exists one.

A couple of years later, Paturi, Pudl\'ak, Saks, and Zane~\cite{ppsz}
came up with a simple but powerful idea. In a preprocessing step, they
apply a restricted version of resolution. This increases the number of
unit clauses the algorithm encounters and therefore increases its
success probability.
This gives an algorithm called $\PPSZ$.  If $F$ has a unique
satisfying assignment, its success probability is quite good (for
$3$-SAT, it is $\Omega(1.308^{-n})$), and the analysis is highly
elegant. The case of multiple satisfying assignments appears to be
much more difficult and has been the subject of several papers so far.
Iwama and Tamaki~\cite{it04} made a major step forward when they
observed that while the success probability of $\PPSZ$ deteriorates as
the number of satisfying assignments increases, that of Sch\"oning's
random walk algorithm~\cite{schoning1999} improves.  They quantified
this tradeoff and obtained an algorithm with a success probability of
$\Omega(1.32373^{-n})$\footnote{Using the new version of~\cite{ppsz} immediately gives the bound $\Omega(1.32267^{-n})$, as stated in~\cite{rolf2006}.}. We denote this combined algorithm, consisting of one run
of $\PPSZ$ and one run of Sch\"oning's random walk algorithm, by $\Comb$.\paragraphprev
%

\paragraphdef{The PPSZ paper.}
There are two versions of~\cite{ppsz}, which we call the old version
and the new version. For unique $k$-SAT, both are the same, but for general
$k$-SAT, the old version of~\cite{ppsz} gives a more complicated
analysis. The old version gives a better bound for 3-SAT and the new version gives a better bound for 4-SAT\@.

Only the new version is published, but the old version is still
available at the Citeseer cache%
\footnote{\htmladdnormallink{http://citeseerx.ist.psu.edu/viewdoc/summary?doi=10.1.1.41.1134}{http://citeseerx.ist.psu.edu/viewdoc/summary?doi=10.1.1.41.1134}}.
However, we have found some minor errors in that version.  There is also a
conference version~\cite{ppszconference} stating the results of the old
version of~\cite{ppsz}, but without most proofs. 
Rolf~\cite{rolf2006} improved the analysis of the old version to get a bound of $\Omega(1.32216^n)$. However~\cite{rolf2006} does not consider 4-SAT\@. We use the ideas of~\cite{rolf2006} for our improvement
of 4-SAT\@.
In Timon Hertli's
master thesis~\cite{thesis}, the old version of~\cite{ppsz} with the
result of~\cite{rolf2006} is presented in a self-contained way. We will reference that thesis for detailed proofs.

\subsection{Our Contribution}

Let $F$ be a satisfiable CNF formula over $n$ variables and $x$ be a variable therein.
We call $x$ {\em critical} if all satisfying assignments of $F$ agree
on $x$. Equivalently, $x$ is critical if exactly one of the formulas
$F^{[x \mapsto 1]}$ and $F^{[x \mapsto 0]}$ is satisfiable. We denote by $c(F)$ the fraction of
critical variables, i.e., the number of critical variables divided by
$n$; if $n=0$, we define $c(F):=1$.

Our contribution consists of two statements: Theorem~\ref{theorem-critical} shows that for our purposes we only need to consider formulas with many critical variables. Point \ref{lemma-item-3} of Lemma~\ref{lemma-forced-critical} then implies that the success
probability of $\PPSZ$ increases if $F$ has many critical variables. This is obtained by slightly modifying the existing analysis of~\cite{ppsz} and~\cite{rolf2006} by taking critical variables into account. However, Lemma~\ref{lemma-forced-critical} is somewhat technical and we need to embed it into a review of the existing analysis.
Theorem \ref{theorem-critical} is very simple, so we state it here:
\begin{theorem}
\label{theorem-critical}
Let $p,q,c^*\in [0,1]$ and $a,b\geq 1$ such that $\frac{q}{b}=\left(1-\frac{c^*}{2}\right)=:r$. Suppose algorithm $\mathcal{A}$ runs in time $a^n 2^{o(n)}$ and for every satisfiable $(\leq k)$-CNF formula $F$ with $c(F)\geq c^*$ finds a satisfying assignment with probability at least $p^n \left(\frac{1}{2}\right)^{o(n)}$. Then there exists an algrotihm $\mathcal{A}'$ that runs in time $\max\{a,b\}^n 2^{o(n)}$ and for every satisfiable $(\leq k)$-CNF formula finds a satisfying assignment with probability at least $\min\{p,q\}^{n}\left(\frac{1}{2}\right)^{o(n)}$.
\end{theorem}
Obviously we can turn $\mathcal{A'}$ into a algorithm that finds a satisfying assignment in expected time $\left(\frac{\max\{a,b\}}{\min\{p,q\}}\right)^n 2^{o(n)}$.
\begin{proof}
By \emph{guessing} $j$ variables we mean fixing in $F$ $j$ variables chosen uniformly at random to values chosen uniformly at random, obtaining the formula $F'$ over at most $n-j$ variables. $\mathcal{A}'$ for each $j\in\{0,\dots,n\}$ repeats the following $b^j$ times: Guess $j$ variables and then run $\mathcal{A}$ on $F'$; the running time bound is trivial. To bound the probability, we first claim that there exists a $j$ such that $a_j\geq \frac{r^j}{n+1}$ where $a_j$ is the probability that after guessing $j$ variables $F'$ is satisfiable and $c(F')\geq c^*$. Suppose this is not the case: Let $b_j$ be the probability that after guessing $j$ variables $F'$ is satisfiable and $c(F')<c^*$. Clearly $a_0+b_0=1$ since $F$ is satisfiable, and $a_{i+1}+b_{i+1}\geq b_i\cdot r$, as guessing one variable preserves satisfiability with probability at least $\left(1-\frac{c^*}{2}\right)=r$. By the assumption, $b_i\cdot r\geq  \left(a_i+b_i-\frac{r^i}{n+1}\right)\cdot r$; from this it is easy to show that $a_n+b_n\geq  r^n-n\frac{r^n}{n+1}=\frac{r^n}{n+1}$. If $j=n$, we have $c(F')=1$ by definition; hence $b_n=0$ and $a_n\geq \frac{r^n}{n+1}$, a contradiction. Now let $j^*$ be the $j$ given by the claim; we repeat $b^{j^*}$ times an algorithm that has success probability at least $\frac{r^{j^*}}{n+1}p^{n-j^*}\left(\frac{1}{2}\right)^{o(n)}$; as $r\cdot b=q$ this gives by a routine argument an algorithm with success probability at least $p^{n-j^*}q^{j^*}\left(\frac{1}{2}\right)^{o(n)}$.
\end{proof}

We improve the analysis for $\PPSZ$ for formulas with many critical
variables. In combination with Theorem~\ref{theorem-critical}, this
gives a success probability of $\Omega\!\left(1.32153^{-n}\right)$
for $3$-SAT and $\Omega\!\left(1.46928^{-n}\right)$ for $4$-SAT\@.
Very recently, Iwama, Seto, Takai, and Tamaki~\cite{istt10} showed
how to combine an improved version of Sch\"oning's
algorithm~\cite{hofmeister02,bs03} with $\PPSZ$ and achieved expected running time of $O(1.32113^{n})$. We combine our improvement with
theirs to obtain a bound of $O(1.32065^{n})$.
\introductionAppendixExplanation

We analyze the algorithm $\Comb(F)$, where $F$ is a CNF formula.
$\Comb$ consists essentially of a call to $\PPSZ$~\cite{ppsz} and to
$\Schoening$~\cite{schoning1999}. In~\cite{it04} it was shown that
$\Comb$ has a better success probability than what the analysis of $\PPSZ$ and
$\Schoening$ gives. Let $\ISTT$ be the algorithm of~\cite{istt10} that improves $\Comb$.
\begin{theorem}
\label{p.t.3sat}
\sloppypar{
There exists an algorithm that for every satisfiable $3$-CNF formula finds a satisfying assignment
with probability
$\Omega\!\left(1.32153^{-n}\right)$
and runs in subexponential time.
}
\end{theorem}
\begin{theorem}
\label{guess-istt}
\sloppypar{
There exists an algorithm that for every satisfiable $3$-CNF formula finds a satisfying assignment
with expected running time
$O\!\left(1.32065^{n}\right)$.
}
\end{theorem}
\theoremAppendixExplanation
\begin{theorem}
\label{guess-istt-weak}
\sloppypar{
There exists an algorithm that for every satisfiable $3$-CNF formula finds a satisfying assignment
with expected running time
$O\!\left(1.321^{n}\right)$.
}\end{theorem}

\begin{theorem}
\label{p.t.4sat}
\sloppypar{
There exists an algorithm that for every satisfiable $3$-CNF formula finds a satisfying assignment
with probability
$\Omega\!\left(1.46928^{-n}\right)$
and runs in subexponential time.
}
\end{theorem}
\sloppypar{
This is already very close to unique $4$-SAT, which has a success probability of
$\Omega\!\left(1.46899^{-n}\right)$. }
The benefit of Theorem~\ref{theorem-critical} is that when proving
Theorems~\ref{p.t.3sat} and~\ref{p.t.4sat}, we only need to consider
formulas with many critical variables.  For example, to prove
Theorem~\ref{p.t.3sat}, we choose $c^*$ such that $1-c^*/2 =
1/1.32153$, i.e., $c^* \approx 0.4866$. Then we have to bound from
below the success probability of $\Comb$ for $3$-CNF formulas $F$ with
$c(F) \geq c^*$.
\subsection{Notation}
We use the notational framework introduced in \cite{welzl05}. We
assume an infinite supply of propositional \emph{variables}. A
\emph{literal} $u$ is a variable $x$ or a complemented variable $\bar
x$. A finite set $C$ of literals over pairwise distinct variables is
called a \emph{clause} and a finite set of clauses is a formula
in \emph{CNF} (Conjunctive Normal Form). We say that a variable $x$
\emph{occurs} in a clause $C$ if either $x$ or $\bar x$ are contained
in it and that $x$ occurs in the formula $F$ if there is any clause
where it occurs. 
We write $\mbox{vbl}(C)$ or $\mbox{vbl}(F)$ to denote
the set of variables that occur in $C$ or in $F$, respectively. 
A clause containing exactly one literal is called a \emph{unit clause}. We
say that $F$ is a 
$(\le k)$\emph{-CNF} formula
if every clause has
size at most $k$.  Let such an $F$ be given and write
$V:=\mbox{vbl}(F)$ and $n:=|V|$.

A \emph{assignment} is a function $\alpha : V \rightarrow
\{0,1\}$ which assigns a Boolean value to each variable. A literal
$u=x$ (or $u=\bar x$) is \emph{satisfied by} $\alpha$ if $\alpha(x)=1$
(or $\alpha(x)=0$). A clause is \emph{satisfied by} $\alpha$ if it
contains a satisfied literal and a formula is \emph{satisfied by}
$\alpha$ if all of its clauses are. A formula is \emph{satisfiable} if
there exists a satisfying truth assignment to its variables.

For an assignment $\alpha$ on $V$ and a set $W\subseteq V$, we denote
by $\alpha\oplus W$ the assignment that corresponds to $\alpha$ on
variables of $V\setminus W$ and is flipped on variables of $W$.

Given a CNF formula $F$, we denote by $\sat(F)$ the set of assignments that
satisfy $F$.

Formulas can be manipulated by permanently assigning values to
variables. If $F$ is a given CNF formula and $x \in \mbox{vbl}(F)$
then assigning $x \mapsto 1$ satisfies all clauses containing $x$
(irrespective of what values the other variables in those closes are
possibly assigned later) whilst it truncates all clauses containing
$\bar x$ to their remaining literals. 

We will write $F^{[x\mapsto1]}$ (and analogously $F^{[x\mapsto 0]}$) to
denote the formula arising from doing just this.

We say that two clauses $C_1$ and $C_2$ conflict on a variable $x$ if
one of them contains $x$ and the other $\ol{x}$. We call $C_1$ and
$C_2$ a resolvable pair if they conflict in exactly one variable $x$,
and we define their \emph{resolvent} by $R(C_1,C_2):=(C_1 \cup C_2)
\setminus \{x,\ol{x}\}$. It is easy to see that if $F$ contains a
resolvable pair $C_1$, $C_2$, then $\sat(F)=\sat(F\cup \{R(C_1,C_2)\})$. A
resolvable pair $C_1$, $C_2$ is $s$-bounded if $|C_1|\leq s$,
$|C_2|\leq s$, and $|R(C_1,C_2)|\leq s$.

By $\Resolve(F,s)$, we denote the set of clauses $C$ that have an
$s$-bounded resolution deduction from $F$. By a straightforward
algorithm, we can compute $\Resolve(F,s)$ in time
$O\left(n^{3s}\poly\left(n\right)\right)$~\cite{ppsz}.

By choosing an element u.a.r.\ from a finite set, we mean choosing it
uniformly at random. By choosing an element u.a.r. from an closed real
interval, we mean choosing it according to the continuous uniform
distribution over this interval. Unless otherwise stated, all random
choices are mutually independent.

We denote by $\log$ the logarithm to the base 2. For the logarithm to
the base $e$, we write $\ln$. We define $0\log 0:= 0$.

\section{Proof of the Main Theorems}

\begin{algorithm}[t]
\caption{$\PPSZ($CNF formula $F$, assignment $\beta$, permutation $\pi)$}
\begin{algorithmic}
\STATE Let $\alpha$ be a 
partial assignment over $\vbl(F)$, initially the empty assignment.
\STATE $G\leftarrow \Resolve(F,\log(|\vbl(F)|))$
\FOR {all $x \in \vbl(G)$, according to $\pi$}
\IF {$\{x\} \in G$}
\STATE $\alpha(x)\leftarrow 1$
\ELSIF {$\{\overline{x}\}\in G$}
\STATE $\alpha(x)\leftarrow 0$
\ELSE 
\STATE $\alpha(x)\leftarrow\beta(x)$
\ENDIF
\STATE $G\leftarrow G^{[x\mapsto \alpha(x)]}$
\ENDFOR
\RETURN $\alpha$
\end{algorithmic}
\end{algorithm}
\begin{algorithm}[t]
\caption{$\PPSZ($CNF formula $F)$}
\begin{algorithmic}
  \STATE \COMMENT{this algorithm is used for 4-SAT} \STATE Choose
  $\beta$ u.a.r.\ from all assignments on $\vbl(F)$ \STATE Choose
  $\pi$ u.a.r.\ from all permutations of $\vbl(F)$ \RETURN
  $\PPSZ(F,\beta, \pi)$
\end{algorithmic}
\end{algorithm}
\begin{algorithm}[t]
\caption{$\Schoening($CNF formula $F,$ assignment $\beta)$}
\begin{algorithmic}
\FOR {$3 |\vbl(F)|$ steps}
\IF {$\beta$ satisfies $F$}
\RETURN $\beta$
\ENDIF
\STATE Select an arbitrary $C\in F$ not satisfied by $\beta$
\STATE Select a variable $x$ u.a.r.\ from $\vbl(C)$ and flip $x$ in $\beta$
\ENDFOR
\RETURN $\beta$
\end{algorithmic}
\end{algorithm}
\begin{algorithm}[t]
\caption{$\Comb($CNF formula $F)$}
\begin{algorithmic}
\STATE
\COMMENT{this algorithm is used for 3-SAT}
\STATE Choose $\beta$ u.a.r.\ from all assignments on $\vbl(F)$
\STATE $\alpha\leftarrow \PPSZ(F,\beta)$
\IF {$\alpha\not\in\sat(F)$}
\STATE $\alpha\leftarrow \Schoening(F,\beta)$
\ENDIF
\RETURN $\alpha$
\end{algorithmic}
\end{algorithm}

In the following let $k\geq 3$ be a fixed integer. Let $F$ be a
satisfiable $(\leq k)$-CNF formula, $V:=\vbl(F)$ and $n:=|V|$. 
We first give the concepts from~\cite{ppsz} needed to understand Lemma~\ref{lemma-forced-critical}. Then we state the lemma and
use it to improve
the bounds on the success probability of $\Comb$ and $\ISTT$ given sufficiently many critical variables. In Section~\ref{section-prooflemma8}, we prove Lemma~\ref{lemma-forced-critical} and also consider 4-SAT. Most concepts used in the proof are from~\cite{ppsz,rolf2006}. Our contribution is to exploit what these concepts yield for critical variables.\paragraphprev

\paragraphdef{Subcubes.}
 For $D\subseteq V$ and $\alpha\in\{0,1\}^V$, the
set $ B(D,\alpha) := \{\beta \in \{0,1\}^V \ | \ \alpha(x)=\beta(x)\
\forall x \in D \} $ is called a \emph{subcube}. The variables in $D$
are called \emph{defining} variables and those in $V\setminus D$
\emph{nondefining} variables. The subcube $B(D,\beta)$ has dimension
$|V\setminus D|$.  For example, if $V = \{x_1,x_2,x_3\}$, $D =
\{x_1,x_3\}$ and $\alpha = (1,0,0)$, then $B(D,\alpha)$ contains
exactly the two assignments $(1,0,0)$ and $(1,1,0)$. Given a nonempty
set $S \subseteq \{0,1\}^V$, there is a partition
$$
\{0,1\}^V = \bigcup_{\alpha \in S} B_{\alpha}
$$
where the $B_\alpha$ are pairwise disjoint subcubes, and $\alpha \in
B_\alpha$ for all $\alpha \in S$. See~\cite{ppsz} for a proof.  For
the rest of the paper, we fix such a partition for $S$ being the set
of satisfying assignments. To estimate the success probability of
$\Comb$, consider the assignment $\beta$ that $\Comb$ chooses
uniformly at random from $\{0,1\}^V$.
\begin{eqnarray*}
\Pr[\Comb(F) \in \sat(F)] & 
= & \sum_{\alpha \in \sat(F)}\Pr[\Comb(F) \in \sat(F) | \ \beta \in B_\alpha]
\cdot \Pr[\beta \in B_\alpha] \\
& \geq & 
\min_{\alpha \in \sat(F)}\Pr[\Comb(F) \in \sat(F) \ | \ \beta \in B_\alpha].
\end{eqnarray*}
Hence instead of analyzing $\Comb$ for an
assignment $\beta$ sampled uniformly at random from all assignments, we fix $\alpha\in\sat(F)$ arbitrarily and
we think of $\beta$ as being sampled from the subcube $B_\alpha$.  Let
$N_\alpha$ be the set of non-defining variables of this cube, and
$D_\alpha$ the set of defining variables.  Intuitively, if
$B_\alpha$ has small dimension, then $\beta$ is likely to be close to
$\alpha$, thus $\Schoening$ has a better success probability:
\begin{lemma}[\cite{it04}]
  $\Pr[\Schoening(F,\beta) \in \sat(F) \ | \ \beta \in B_\alpha]
  \geq (2-2/k)^{-|N_\alpha|}$.
  \label{lemma-schoening-IT}
\end{lemma}

\paragraphdef{Placements.}
As a next step, we analyze
$\PPSZ(F,\beta,\pi)$ with $\beta$ chosen uniformly at random from
$B_\alpha$ and the permutation also chosen from some subset of
permutations. A {\em placement} of the variables $V$ is a function
$\sigma: V \rightarrow [0,1]$, and a {\em uniform random placement} is defined
by chosing $\sigma(x)$ uniformly at random from $[0,1]$ independently
for each $x \in V$.  With probability $1$, a uniform random placement is
injective and gives rise to a uniformly distributed permutation via
the natural ordering $<$ on $[0,1]$. For the rest of the paper, we
will view $\pi$ as a placement rather than a permutation. Let $\Gamma$
be a measurable set of placements.  Then
\begin{multline*}
  \Pr[\PPSZ(F,\beta,\pi)\in\sat(F) \ | \ \beta \in B_\alpha] \geq  \\
  \Pr[\PPSZ(F,\beta,\pi)\in\sat(F) \ | \ \beta \in B_\alpha, \pi \in \Gamma] 
  \cdot \Pr[\pi \in \Gamma].
\end{multline*}
The benefit of this is that we can tailor $\Gamma$ towards our needs,
i.e., making the conditional probability $\Pr[\PPSZ(F,\beta,\pi) \in\sat(F) \ | \
\beta \in B_\alpha, \pi \in \Gamma]$ fairly large. This may come at
the cost of making $\Pr[\pi \in
\Gamma]$ small.\\

\paragraphdef{Forced variables.}
Suppose the permutation $\pi$ orders the
variables $V$ as $(x_1,\dots,x_n)$. Let $\alpha$ be a satisfying
assignment of $F$. Imagine we call $\PPSZ(F,\alpha,\pi)$. The
algorithm applies bounded resolution to $F$, obtaining $G =
\Resolve(F,\log(n))$ and sets the variables $x_1,\dots,x_n$
step by step to their respective values under $\alpha$, creating a
sequence of formulas by $G=G_0,G_1,\dots,G_n$, where $G_{i} =
G_{i-1}^{[x_i \mapsto \alpha(x_i)]}$ for $1 \leq i \leq n$. Since
$\alpha$ is a satisfying assignment, $G_n$ is the empty formula. We
say $x_i$ is {\em forced} with respect to $\alpha$ and $\pi$ if
$G_{i-1}$ contains the unit clause $\{x_i\}$ or $\{\bar{x}_i\}$. By
$\forced(\alpha,\pi)$ we denote the set of variables $x$ that are
forced with respect to $\alpha$ and $\pi$. If $x$ is not forced, we
say it is {\em guessed}. We denote by $\guessed(\alpha,\pi)$ the set
of guessed variables. Note that $\PPSZ(F,\beta,\pi)$ returns $\alpha$
if and only if $\alpha(x) = \beta(x)$ for all $x \in
\guessed(\alpha,\pi)$.  Furthermore, since $\beta$ is chosen uniformly
at random from $B_\alpha$, we already have $\alpha(x)=\beta(x)$ for
all $x \in D_\alpha$. Therefore
\begin{eqnarray}
	\Pr[\PPSZ(F,\beta,\pi)\in\sat(F)] & \geq & 
  \Pr[\PPSZ(F,\beta,\pi)=\alpha] \\ & = & 
  \E\left[2^{-|N_\alpha \cap \guessed(\alpha,\pi)|}\right] \nonumber \\
  & \geq &
  2^{-\E[|N_\alpha \cap \guessed(\alpha,\pi)|]},\label{ineq-success}
\end{eqnarray}
where the inequality comes from Jensen's inequality applied to the
convex function $t \mapsto 2^{-t}$. Note that (\ref{ineq-success})
holds when taking $\pi$ uniformly at random as well as when sampling
it from some set $\Gamma$. Using linearity of expectation, we see that
\begin{eqnarray}
\E[|N_\alpha \cap \guessed(\alpha,\pi)|] = 
\sum_{x \in N_\alpha}\Pr[x \in \guessed(\alpha,\pi)].
\label{sum-of-unforced}
\end{eqnarray}
Now if $\alpha$ is the unique satisfying assignment, then $N_\alpha =
V$. For $3$-SAT, one central result of~\cite{ppsz} is that
\begin{lemma}[\cite{ppsz}]
  Let $F$ be a satisfiable $3$-CNF formula with a unique satisfying
  assignment $\alpha$. Then for every $x \in \vbl(F)$, it holds that
  $\Pr[x \in \guessed(\alpha,\pi)] \leq 2\ln(2)-1+o(1) < 0.3863$.
\label{lemma-unique-forced}
\end{lemma}

Combining the lemma with (\ref{ineq-success}) shows that $\PPSZ$ on
$3$-CNF formulas with a unique satisfying assignment has a success
probability of at least $2^{-(2\ln(2)-1+o(1))n} \in
\Omega(1.308^{-n})$. For the case of multiple satisfying assignments,
the lemma does not hold anymore.\paragraphprev

\paragraphdef{Critical variables.}
Let $F$ be a satisfiable CNF formula and
$x$ a variable. Recall that we call $x$ {\em critical} if all satisfying assignments
of $F$ agree on $x$. The following observation is not difficult
to show:
\begin{observation}
  Let $F$ be a satisfiable CNF formula and let $V_C$ be the set of
  critical variables. Let $B_\alpha$ be the subcube as defined
  above. For a satisfying assignment $\alpha$, let $N_\alpha$ be the
  set of nondefining variables. 
  Then $V_C \subseteq N_\alpha$.
\end{observation}
\begin{lemma}
  Let $F$ be a satisfiable $3$-CNF formula and $\alpha$ be a
  satisfying assignment. There is a measurable set $\Gamma \subseteq
  [0,1]^V$ of placements such that for $\beta = \valueBetaH$ and
  $\gamma = \valueGammaH$, we have
  \begin{enumerate}
  \item $\Pr[\pi \in \Gamma] \geq 2^{-\beta |D_\alpha|-o(n)} \approx
    \valueOneHalfToTheBetaH^{|D_\alpha|-o(n)}$,
  \item $\Pr[x \in \forced(\alpha,\pi) \ | \ \pi \in \Gamma] \geq
    \gamma - o(1) \approx \valueGammaH - o(1)$ for all $x \in
    N_\alpha$,
  \item $\Pr[x \in \forced(\alpha,\pi) \ | \ \pi \in \Gamma] \geq
    2-2\ln(2) - o(1) \approx 0.6137056$ for all critical $x\in V$.
      \label{lemma-item-3}
    \end{enumerate}
    \label{lemma-forced-critical}
\end{lemma}
The important part of the lemma is point \ref{lemma-item-3}, namely
that critical variables are forced with a larger probability than
non-critical ones.
\begin{proof}[Proof of Theorem~\ref{p.t.3sat}]
Using Theorem~\ref{theorem-critical}, we can assume $c(F)\geq 0.48659459$.
Let $\Delta := |D_\alpha|/|V|=1-|N_\alpha|/|V|$ be the fraction of defining
  variables.
  Combining (\ref{sum-of-unforced}) with Lemma~\ref{lemma-forced-critical},
  we obtain
  \begin{eqnarray*}
    \E[|N_\alpha \cap \guessed(\alpha,\pi)| \ | \ 
    \pi \in \Gamma] & = & 
    \sum_{x \in N_\alpha}\Pr[x  \in \guessed(\alpha,\pi)] \\
    & \leq &
    (2\ln 2 - 1) |V_C| + (1-\gamma) |N_\alpha \setminus V_C| + o(n) \\
    & \leq & (2\ln 2 - 1) c^* n +  (1-\gamma) (1-\Delta - c^*) n + o(n) \\
    & = & 0.389532 n - \valueOneMinusGammaH \Delta n + o(n).
  \end{eqnarray*}
  The expected fraction of nondefining variables we have to guess is
  thus a little bit larger than in the case of a unique satisfying
  assignment, where it is $\approx 0.3863$. Together with
  (\ref{ineq-success}), we conclude that the success probability
  of $\PPSZ$ is at least
  \begin{eqnarray}
  \Pr[\PPSZ(F,\beta,\pi)=\alpha \ | \ \beta \in B_\alpha] & \geq & 
  \Pr[\PPSZ(F,\beta,\pi)=\alpha \ | \ \beta\in B_\alpha, \pi \in \Gamma]
  \cdot \Pr[\pi \in \Gamma] \nonumber\\
  & \geq & 2^{-\E[|N_\alpha \cap \guessed(\alpha,\pi)| \ | \ 
    \pi \in \Gamma]}\cdot \Pr[\pi \in \Gamma] \nonumber\\
  & \geq & 2^{-0.389532 n + \valueOneMinusGammaH \Delta n} \cdot 
  \valueOneHalfToTheBetaH^{\Delta n} \cdot 2^{-o(n)} \nonumber\\
  & \geq & 1.3099684^{-n} \cdot 1.328369^{-\Delta n} \cdot 2^{-o(n)}.
  \label{eq-ppsz-prob}
  \end{eqnarray}
  Our bound on the success probability of $\PPSZ$ thus deteriorates with
  the number of defining variables. A bigger subcube $B_\alpha$ is
  better for $\PPSZ$. We combine this with the bound for Sch\"oning's
  algorithm from Iwama and Tamaki~\cite{it04}, stated 
  above in Lemma~\ref{lemma-schoening-IT}
  \begin{equation}
  \Pr[\Schoening(F,\beta) \in \sat(F) \ | \ \beta \in B_\alpha]
  \geq (2-2/k)^{-(1-\Delta)n}.
  \label{eq-schoening-prob}
  \end{equation}
  \sloppypar{
  The combined worst case is with $\Delta \approx
  0.0309273$, in which case both (\ref{eq-ppsz-prob}) and (\ref{eq-schoening-prob}) evaluate to $\Omega(1.32153^{-n})$. 
   Therefore for any $\Delta$, at least one of $\Schoening$
  and $\PPSZ$ has a success probability of $\Omega(1.32153^{-n})$.
  }
  \end{proof}
\begin{proof}[Proof of Theorem \ref{guess-istt-weak}]
Lemma 6 from~\cite{istt10} tells us that there is an algorithm $\ISTTSch$ that improves $\Schoening$ such that for all $m^*\in [0,\frac{1}{3}]$ we have, after preprocessing time $6^{m^* n}$,
\[
\Pr[\ISTTSch(F,\beta)\in\sat(F) \ | \ \beta \in B_\alpha]
\geq 1.012795^{m^* \cdot n}\cdot 1.2845745^{\Delta n}\cdot (3/4)^n.
\]
We want to prove that by replacing $\Schoening$ with $\ISTTSch$ in $\Comb$, we obtain expected running time of
  $O(1.321^{n})$. Setting $c^* := 0.48599$ and $m^*:=0.155371873$ gives $1-c^*/2\geq 1/1.321$ and $6^{m^*}\geq 1.321$. With this choice of $c^*$, we have the following bound for $\PPSZ$ (obtained as in the previous proof, but with a different constant $c^*$):
\[
\Pr[\PPSZ(F,\beta,\pi)=\alpha \ | \ \beta \in B_\alpha]
\geq 1.31^{-n} \cdot 1.3312^{-\Delta n} \cdot 2^{-o(n)}. 
\]
The combined worst case is at $\Delta\approx 0.029225$ where
$1.31^{-n} \cdot 1.3312^{-\Delta n}> 1.321^{-n}$ and
$1.012795^{m^*\cdot n}\cdot 1.2845745^{\Delta n}\cdot (3/4)^n> 1.321^{-n}$, proving that the combined success probability is $\Omega(1.321^{-n})$ (after preprocessing time $O(1.321^n)$).
\end{proof}

%
\section{Proof of Lemma~\ref{lemma-forced-critical}}
\label{section-prooflemma8}
\subsection{Critical Clause Trees}
 Let
$G:=\Resolve(F,\log(n))$. Note that $\vbl(F)=\vbl(G)$ and
$\sat(F)=\sat(G)$.
A \emph{critical clause} for $x\in V$ w.r.t.\  $\alpha$ is a clause where $\alpha$ satisfies exactly one literal and this literal is over $x$. It can be easily seen that if the output of $\PPSZ$ should be $\alpha$, then exactly the critical clauses of $G$ are the clauses that might turn into unit clauses. Note that the \emph{defining} variables are assumed to be set correctly, so we only need to consider critical clauses for \emph{nondefining} variables here.

We now define critical clause trees, a concept that tells us which critical clauses we can expect in a CNF formula after bounded resolution. Let $T$ be a rooted tree in which every node is either labeled with a variable from $V$ or is unlabeled. A \emph{cut} in a rooted tree is a set of nodes $A$ such that the root is not in $A$ and every path from the root to a leaf contains at least one node in $A$. The \emph{depth} of a node is the distance to the root.  
For a set $A$ of nodes, $\vbl(A)$ denotes the set of variables occurring as labels in $A$.
We say $T$ is a \emph{critical clause tree} for $x$ w.r.t.\ $G$ and $\alpha$ if the following properties hold:

\begin{enumerate}
\item The root is labeled by $x$.
\item On any path from the root to a leaf, no two nodes have the same label.
\item For any cut $A$ of the tree, there is a critical clause $C\in G$ w.r.t.\ $\alpha$ where the satisfied literal is over $x$ and every unsatisfied literal is over some variable in $\vbl(A)$.
\end{enumerate}

\begin{wrapfigure}{r}{4cm}
\includegraphics{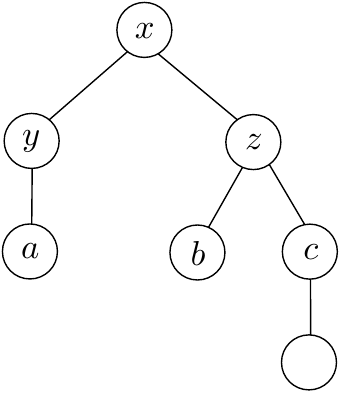}
\caption{Example Critical Clause Tree}
\label{f.treeexample}
\end{wrapfigure}

It is shown in~\cite{ppsz} that we can construct a critical clause tree for $x\in N_\alpha$ as follows: Start with the root labeled $x$. Now we can repeatedly extend a leaf node $v$. Let $L$ be the set of labels that occur on the path from $v$ to the root. If $\alpha\oplus L$ does not satisfy $F$, then we can extend the tree at that node: There is a clause $C$ in $F$ (not in $G$) not satisfied by $\alpha\oplus L$. For each literal in $C$ that is not satisfied by $\alpha$, we add a child to $v$ labeled with the variable of that literal. If there are no such literals, we add an unlabeled node. As clauses of $F$ have at most $k$ literals, each node has at most $k-1$ children. If the constructed tree has at most $\log(n)$ nodes (as we do $\log(n)$-bounded resolution), then it is a critical clause tree for $x$ w.r.t.\ $G$ and $\alpha$.

We give a simple example: Let \[F:=\{\{x,\ol{y},\ol{z}\},\{x,y,\ol{a}\},\{z,\ol{b},\ol{c}\},\{x,z,c\}\}.\] For the all-one assignment and $x$, we can get the tree shown in Figure \ref{f.treeexample} by the described procedure. $\{a,b\}$ is a cut in this tree. We have $R(\{z,\ol{b},\ol{c}\},\{x,z,c\})=\{x,z,\ol{b}\}$, $R(\{x,\ol{y},\ol{z}\},\{x,y,\ol{a}\})=\{x,\ol{z},\ol{a}\}$ and $R(\{x,z,\ol{b}\},\{x,\ol{z},\ol{a}\})=\{x,\ol{a},\ol{b}\}$, giving the required critical clause.

If $\alpha$ is the only satisfying assignment of $F$, $\alpha\oplus L$
never satisfies $F$, and we can build a tree where all leafs are at
depth $d:=\floor{\log_k(\log(n))}$. We call this a {\em full tree}. The
important observation is now that this also works if $x$ is a
\emph{critical variable}, as in that case $\alpha\oplus L$ also never
satisfies $F$, as $x\in L$.

In the general case, however, the assignment $\alpha\oplus L$ might satisfy $F$ so that we cannot extend the tree. However if $L$ consists only of \emph{nondefining} variables, then we know that $\alpha\oplus L$ does not satisfy $F$. Hence we can get a tree where every leaf not at depth $d$ is labeled by a \emph{defining} variable. We define the trees $T_x$ we will use in the analysis:
\begin{definition}
\label{p.d.trees}
For $x\in N_\alpha$, construct the critical clause tree for $x$ as
follows: If $x$ is a critical variable, then construct $T_x$ such that
all leaves are at depth $d$, i.e., construct a full tree. Otherwise,
construct $T_x$ such that all leaves not labeled by defining variables
are at depth $d$.
\end{definition}

This means that a tree might just consist of a root where all children are labeled with defining variables, which essentially nullifies the benefits from resolution. To cope with this, we have to make defining variables more likely to occur at the beginning. We achieve this by choosing the set $\Gamma$ of placements whose existence we claim in Lemma \ref{lemma-forced-critical} in a way such that exactly that happens. 

\begin{definition}
  A function $H:[0,1]\to [0,1]$ is called a \emph{nice distribution
    function} if $H$ is non-decreasing, uniformly continuous,
  $H(0)=0$, $H(1)=1$, $H$ is differentiable except for finitely many
  points and $H(r)\geq r$.
\end{definition}
Compared with~\cite{ppsz}, we added the requirement $H(r)\geq r$. This will
mean that defining variables cannot be less likely to occur at the
beginning than nondefining variables. We now define a random placement where defining
variables are placed with distribution function $H$:
\begin{definition}
  Let $H$ be a nice distribution function. By $\pi_H$, we define the
  random placement on $V$ s.t.\ $\pi(x)$ for $x\in N_\alpha$ is u.a.r.\
  $\in [0,1]$, and for $x\in D_\alpha$ and $r\in[0,1]$, $\Pr(\pi(x)\leq
  r)=H(r)$.
  \label{nice-random-placement}
\end{definition}

Assume that the variables are processed according to some placement
$\pi$. Consider $T_x$. If there is a cut $A$ such that $\pi(y)<\pi(x)$ for every $y\in\vbl(A)$, 
then $x$ is forced, as the
corresponding critical clause has turned into a unit clause for $x$.
Denote the probability that $S_x(\pi)$ is a cut in $T_x$ by
$Q(T_x,\pi)$.

For $r\in[0,1]$, let $R_k(r)$ be the smallest non-negative $x$ that
satisfies $x=(r+(1-r)x)^{k-1}$ and $R_k:=\int_0^1 R_k(r)dr$. It was
shown in~\cite{ppsz} that if $T_x$ is a full tree, then
\[Q(T_x,\pi_U)\geq R_k-o(1).\]
$R_k(r)$ can be understood as follows: Take an infinite $(k-1)$-ary tree and mark each node as ``dead'' with probability $r$, except the root. $R_k(r)$ is the probability that this tree contains an infinite path that starts at the root and contains only ``alive'' nodes.

We have $R_3=2-2\ln 2\approx0.6137$ and $R_4\approx 0.4451$. For $r\in[0,\frac{1}{2}]$, we have $R_3(r)=\left(\frac{r}{1-r}\right)^2$ and for $r\in[\frac{1}{2},1]$, we have $R_3(r)=1$. 
As $H(r)\geq r$, and by definition of $\pi_H$ and of a cut, it is obvious that
\begin{eqnarray}
  Q(T_x,\pi_H)\geq R_k-o(1), 
  \label{bound-full}
\end{eqnarray}
 if $T_x$ is a full tree.  If $T_x$
is not a full tree, we do not have any good bounds on $Q(T_x, \pi_U)$.
In~\cite{rolf2006} it is shown that if $T_x$ is not necessarily a full
tree, but a tree in which every leaf not at depth $d$ is labeled by a
defining variable, then
\begin{eqnarray}
  Q(T_x,\pi_H)\geq \gamma_H-o(1),
  \label{bound-not-full}
\end{eqnarray}
where
\[\gamma_H=\int_0^1 \min\{H(r)^{k-1},R_k(r)\}dr.\]
Obviously $\gamma_H\leq R_k$, which means that the bound
(\ref{bound-full}) for full trees is at least as strong as the bound
(\ref{bound-not-full}) for general trees. The $H(r)^{k-1}$ term
corresponds to the tree that consists of a root where all
children are labeled with defining variables and are thus leaves
(remember that there are at most $k-1$ children).  It takes a small
lemma to show that this tree and the full tree are the worst cases.
See~\cite{thesis} for details.  The following observation summarizes
this:
\begin{observation}
\label{p.o.trees}
If $x$ is a \emph{critical} variable, then $Q(T_x,\pi_H)\geq
R_k-o(1).$ If $x$ is a noncritical \emph{nondefining} variable, then
$Q(T_x,\pi_H)\geq \gamma_H-o(1).$
\end{observation}
We want to find a set $\Gamma$ of placements such that a placement
chosen uniformly at random from $\Gamma$ behaves more or less like
$\pi_H$.
\begin{lemma}[old version of \cite{ppsz}]
\label{p.l.Hlemma}
Let $H$ be a nice distribution function. If $|D_\alpha|\geq\sqrt{n}$,
there is a set of placements $\Gamma$ depending on $n$ with the
following properties: Let $\pi_\Gamma$ be the placement choosen
uniformly at random from $\Gamma$. Then for any tree $T$ with at most
$\log(n)$ nodes we have
\[Q(T,\pi_\Gamma)\geq Q(T,\pi_H)-o(1)\]
and
\[Pr(\pi_U\in\Gamma)\geq 2^{-\beta_H |D_\alpha|-o(n)}\]
with
\[\beta_H:=\int_0^1 h(r)\log\left(h(r)\right)dr\]
where $h(r)$ is the derivative of $H(r)$.
\end{lemma}
The proof of this lemma is long and complicated, see \refHtheory{}
in~\cite{thesis}. The case $|D_\alpha|<\sqrt{n}$ is easy to handle:
The probability that all defining variables come at the beginning is substantial, and we are essentially in the
(good) unique case.  

Below we will show how to choose a good function $H$ for the case
$k=3$ and $k=4$. To get an intuition, see Figure~\ref{f.hplot} for a
plot of $H$ for $k=3$. With this function, one obtains $\gamma_H
\approx \valueGammaH$ and $\beta_H \approx \valueBetaH$.
Together with Lemma~\ref{p.l.Hlemma} and Observation~\ref{p.o.trees},
we conclude that for a {\em critical} variable $x$
$$
 \Pr[x \in \forced(\alpha,\pi)] \geq Q(T_x,\pi_H)-o(1) \geq R_k - o(1)
 \geq 0.61371,
$$
and for a non-critical non-defining variable $x$
$$
\Pr[x \in \forced(\alpha,\pi)] \geq Q(T_x,\pi_H) \geq \gamma_H-o(1)
\geq \valueGammaH - o(1).
$$

\subsection{Choosing a good $H$}

Let now $k=3$. We choose $H$ as in~\cite{rolf2006}: Let $\theta\in [0.5,1]$ be a parameter. With some appropriate parameters $a$ and $b>1$, we define $H(r)$ as follows:
\[H(r):=
\begin{cases}
r/\theta&\mathrm{if\ }r\in [0,1-\theta)\\
1-\left(-a \ln(r)\right)^b&\mathrm{if\ }r\in[1-\theta,1]
\end{cases}
\]
\paragraphdef{3-SAT.}
\begin{wrapfigure}{r}{6cm}
\includegraphics[width=5.5cm]{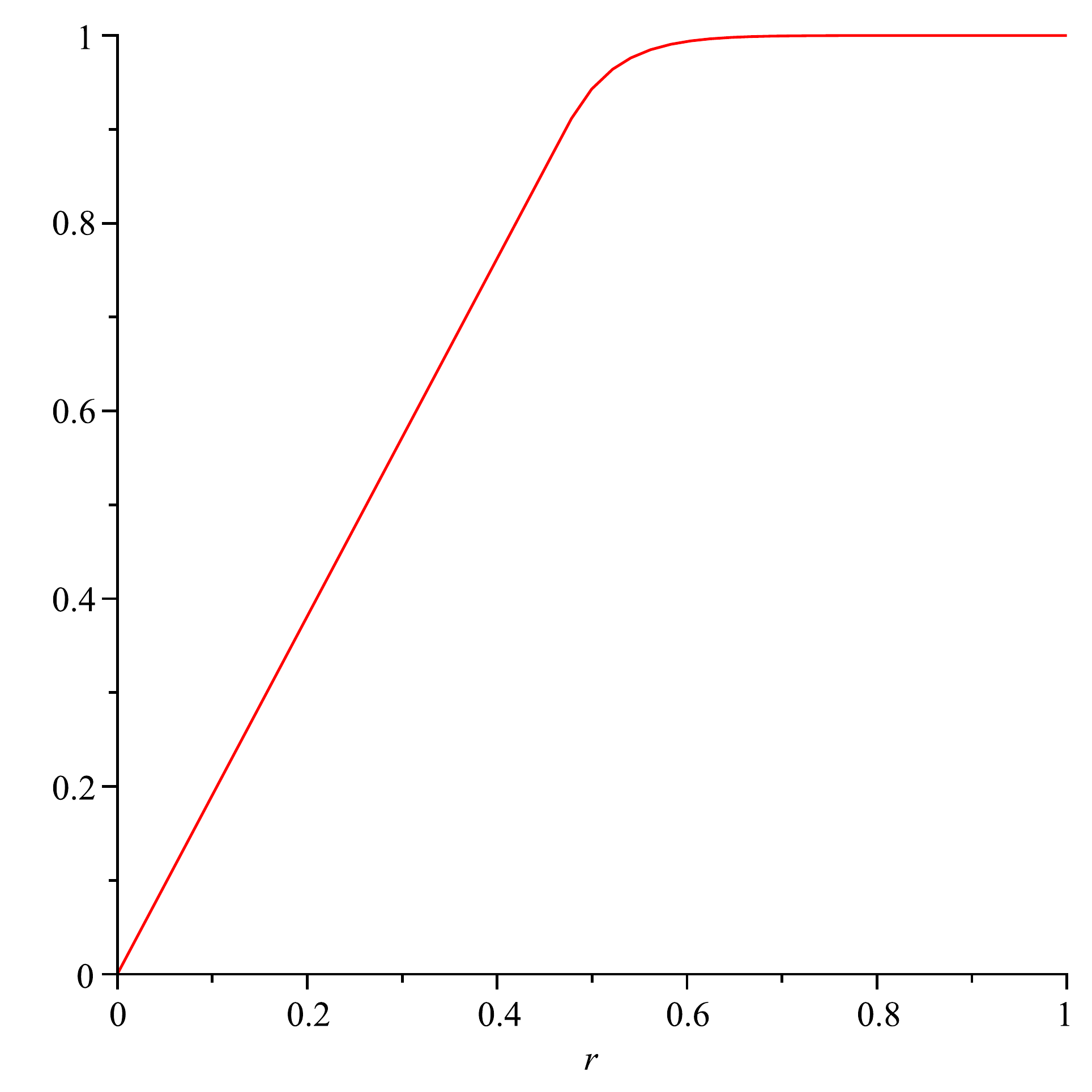}
\caption{$H(r)$ for 3-SAT}
\label{f.hplot}
\end{wrapfigure}
To determine $a$ and $b$, we set the constraints 
\[H(1-\theta)=R_3(1-\theta)^{1/2}\] 
(as $\theta\geq 1/2$, this right-hand side is equal to $\frac{1-\theta}{\theta}$) and 
\[h(1-\theta)=1/\theta.\]

If these constraints are satisfied, $H(r)$ is a nice distribution function that is differentiable on $[0,1]$. Figure \ref{f.hplot} gives a plot of the $H(r)$ we use. Numerical optimization gives $\theta\approx 0.52455825$ and as before $c^*\approx 0.48659459$. See \refHevaluation{} in~\cite{thesis} for details of the computation.
This gives
\[a\approx0.96782885577,\]
\[b\approx7.19709520894,\]
\[\beta_H\leq \valueBetaH,\]
\[\gamma_H\geq \valueGammaH.\]
This concludes the proof of Lemma~\ref{lemma-forced-critical}.\\

\paragraphdef{4-SAT.}
For 4-SAT, we use the $H$ corresponding to the new version of~\cite{ppsz}. For some parameter $\theta\in[\frac{2}{3},1]$, we let $H(r):=\min\{\frac{r}{\theta},1\}$.
It turns out that the optimum is when $\beta_H=1-\gamma_H$. In that case it is easily seen that the bound for \PPSZ{} does not depend on $|D_\alpha|$, and hence we do not need \Schoening{}.
Numerical optimization gives $\theta\approx 0.6803639$ and $c^*\approx 0.63878808$. This implies the success probability $\Omega\!\left(1.46928^{-n}\right)$, proving Theorem \ref{p.t.4sat}.

\section{Conclusion}
We have shown how to improve \PPSZ{} by a preprocessing step that
guarantees that a substantial fraction of variables will be critical. With this,
we were able to improve the bound for 3-SAT and 4-SAT
from~\cite{rolf2006}. We have also shown that our approach nicely
combines with the improvement by~\cite{istt10} by giving an even
better bound. In 4-SAT, we are
already very close to the unique case. We do not know if a more
refined choice of $H$ (similar to~\cite{rolf2006}), possibly depending
on $\Delta$, allows us to close that gap.

It is interesting to see that we could make use of multiple assignments in the guessing step before considering just one assignment using the subcube partition.
\section*{Acknowledgments}
We thank Emo Welzl for many fruitful discussions and continous support and Konstantin Kutzkov for pointing us to~\cite{istt10}.

\appendix
\section{Proof of the $O(1.32065^n)$ bound}
In this section we prove that there exists an algorithm that for every satisfiable $3$-CNF formula finds a satisfying assignment
in expected running time
$O\!\left(1.32065^n\right)$, as stated in Theorem \ref{guess-istt}.

First we show how to derive from~\cite{istt10} a statement similar to Lemma 6 of~\cite{istt10}. They have used such a lemma, but did not state it explicitly. Then analogously to before, we give the parameters $\theta$, $c^*$ and $m^*$ (derived by numerical optimization) to prove the claimed bound.
\begin{lemma}[\cite{istt10}]
Let $f_m:=\frac{64}{63}$ and $f_d:=1.28248358$. Let $\Delta:=|D_\alpha|/|V|$, as before.
For $m^*\in [0,\frac{1}{3}]$ we have after preprocessing time $\mathcal{O}(6^{m^* n})$ that
\[
\Pr[\ISTTSch(F,\beta)\in\sat(F) \ | \ \beta \in B_\alpha]
\geq \left(f_m\right)^{m^* \cdot n}\cdot (f_d)^{\Delta n}\cdot (3/4)^n.
\]
\end{lemma}
Note that $f_m=\frac{64}{63}\approx1.015873>1.012795$, which is corresponding number in Lemma 6 of~\cite{istt10}; however $f_d$ decreases from $1.2845745$ to $1.28248358$. This means that we are better if $\Delta$ is small, but worse if $\Delta$ is large. However, as the combined worst case is for small $\Delta$ ($\approx 0.0286138$) , we improve the probability of the combined algorithm nonetheless.
\begin{proof}
We can interpret $\ISTTSch$ as follows: We first do a preprocessing step using an algorithm from Baumer and Schuler~\cite{bs03} that takes time $\mathcal{O}(6^{m^* n})$. This either finds a satisfying assignment of $F$ with high probability or it finds a set of independent 3-clauses $\mathcal{C}$ (clauses that do not share variables) of size at least $m^*\cdot n$. In the latter case, this set of independent clauses is stored and $\ISTTSch$ does the following: The initial assignment $\beta$ is modified on the variables of $\mathcal{C}$ to an assignment $\beta'$. Then $\Schoening(F,\beta')$ is called.

In~\cite{istt10} it was shown that we can look at each clause in $C\in\mathcal{C}$ independently in terms of the probability of $\Schoening(F,\beta')$. For a satisfying assignment $\alpha$, we determine the \emph{type} of a clause $C$ by the number of literals that correspond to non-defining variables, defining variables as satisfied literals, and defining variables as dissatisfied literals. There are 9 types, which are denoted by $0,10,11,20,21,22,31,32,33$. The first digit denotes the number of defining variables of the literals of $C$, the second digit denotes the number of defining variables corresponding to satisfied literals. The corresponding probability of $\Schoening$ is listed in Table~\ref{t.istt-sch-prob}, as in Table 3 of~\cite{istt10}.

\begin{table}
\begin{tabular}{cc}
type $j$ of $C$&$e(j)$, the $\Schoening$ success probability on the variables of $C$\\
$0$&$\frac{3}{7}$\\
$10$&$\frac{379}{672}$\\
$11$&$\frac{181}{336}$\\
$20$&$\frac{3}{4}$\\
$21$&$\frac{29}{42}$\\
$22$&$\frac{29}{42}$\\
$31$&$1$\\
$32$&$\frac{37}{42}$\\
$33$&$1$\\
\end{tabular}
\caption{Clause type and $\Schoening$ success probability}
\label{t.istt-sch-prob}
\end{table}

Iwama et al.\ have then shown that there are $16$ patterns how the subcube partition (dependent on the independent 3-clauses $\mathcal{C}$) of the assignments on the variables of a clause can result in these types, as shown in Table~\ref{t.istt-pattern}. Note that patterns $9,10$ and patterns $13,14$ have the same type outcomes, but are noted as different patterns in~\cite{istt10}. Pattern $0$ corresponds to type $0$ and it was not treated explicitly as a pattern in~\cite{istt10}.
Furthermore it was shown that with high probability the number of resulting types is close to the expectation.
Let $p(i,j)$ denote the probability that pattern $i$ turns into type $j$. For type $j$, let $d(j)$ denote the number of defining variables (i.e. the first digit). 
Then we have to show the following bound for every pattern $i$:
\[\prod_{j} e(j)^{p(i,j)}\geq \left(\frac{3}{4}\right)^3 f_m \prod_j \left(f_d\right)^{p(i,j)d(j)}.\]
The left-hand side corresponds to the expected $\Schoening$ probability of a clause of pattern $i$; the right-hand side corresponds to the term we want in the statement of the lemma. See~\cite{istt10} for details.
As $f_m$ is a rational number that is easily derived from pattern $0$ and hence type $0$, we can check the following inequality for patterns $1$ to $15$:
\[f_d(j):=\left(\left(\frac{3}{4}\right)^{-3}\frac{1}{f_m}\prod_{j} e(j)^{p(i,j)}\right)^{\frac{1}{\sum_{j}p(i,j)d(j)}}\geq f_d.\]
We have listed the numerical results of $f_d(j)$ in Table \ref{t.istt-fdj} (9 significant digits, rounded down). The worst case for $f_d(j)$ is pattern $4$, which corresponds to $f_d$ of the lemma statement.

\begin{table}
\begin{tabular}{cl}
pattern number&probability distribution of types\\

$0$ &$1:0$\\

$1$ &$\frac{1}{2}:10,\frac{1}{2}:11$\\

$2$ &$\frac{2}{4}:11,\frac{1}{4}:20,\frac{1}{4}:21$\\

$3$ &$\frac{2}{4}:10,\frac{1}{4}:21,\frac{1}{4}:22$\\

$4$ &$\frac{1}{4}:20,\frac{2}{4}:21,\frac{1}{4}:22$\\

$5$ &$\frac{4}{8}:11,\frac{2}{8}:20,\frac{1}{8}:31,\frac{1}{8}:32$\\

$6$ &$\frac{4}{8}:10,\frac{2}{8}:22,\frac{1}{8}:31,\frac{1}{8}:32$\\

$7$ &$\frac{4}{8}:10,\frac{2}{8}:21,\frac{1}{8}:32,\frac{1}{8}:33$\\

$8$ &$\frac{4}{8}:10,\frac{1}{8}:31,\frac{2}{8}:32,\frac{1}{8}:33$\\

$9$ &$\frac{2}{8}:20,\frac{2}{8}:21,\frac{2}{8}:22,\frac{1}{8}:31,\frac{1}{8}:32$\\

$10$&$\frac{2}{8}:20,\frac{2}{8}:21,\frac{2}{8}:22,\frac{1}{8}:31,\frac{1}{8}:32$\\

$11$&$\frac{2}{8}:20,\frac{4}{8}:21,\frac{1}{8}:32,\frac{1}{8}:33$\\

$12$&$\frac{2}{8}:20,\frac{2}{8}:22,\frac{2}{8}:31,\frac{2}{8}:32$\\

$13$&$\frac{2}{8}:20,\frac{2}{8}:21,\frac{1}{8}:31,\frac{2}{8}:32,\frac{1}{8}:33$\\

$14$&$\frac{2}{8}:20,\frac{2}{8}:21,\frac{1}{8}:31,\frac{2}{8}:32,\frac{1}{8}:33$\\

$15$&$\frac{2}{8}:20,\frac{2}{8}:31,\frac{3}{8}:32,\frac{1}{8}:33$\\
\end{tabular}
\caption{Probability distribution of types}
\label{t.istt-pattern}
\end{table}

\begin{table}
\begin{tabular}{cc}
pattern $j$&$f_d(j)$\\
$1$&$1.28611973$\\
$2$&$1.28272221$\\
$3$&$1.28466750$\\
$4$&$1.28248358$\\
$5$&$1.29339711$\\
$6$&$1.29507819$\\
$7$&$1.29507819$\\
$8$&$1.30294154$\\
$9$&$1.29080377$\\
$10$&$1.29080377$\\
$11$&$1.29080377$\\
$12$&$1.29749876$\\
$13$&$1.29749876$\\
$14$&$1.29749876$\\
$15$&$1.30300231$\\
\end{tabular}
\caption{$f_d(j)$ for pattern $j$}
\label{t.istt-fdj}
\end{table}
\end{proof}

Starting from the previous lemma, we now prove Theorem \ref{guess-istt}. We let $\theta:=0.5224565$, $c^*:=2-\frac{2}{1.32065}\approx 0.4855942149$, $m^*:=\log_6(1.32065)\approx 0.155223982$ ($c^*$ and $m^*$ are rounded down). It is easily seen that the choice of $c^*$ and $m^*$ work for the bound we want to achieve. Note that if we would want to have more significant digits in the bound, we would need to lower $c^*$ and $m^*$ slightly.  As before, using the $H$ from~\cite{rolf2006}, we have now
\[a \approx 0.99012456677,\]
\[b \approx 7.85858019246,\]
\[\beta_H \leq 0.8180299645,\]
\[\gamma_H\geq 0.6083696059.\]

We now obtain a lemma analogous to Lemma \ref{lemma-forced-critical} but with different $\beta$ and $\gamma$. It is straightforward to show analogously to before that we get the combined bound of $\Omega(1.32065^{-n})$ for one combined execution by considering the combined worst-case $\Delta\approx 0.0286138$.

\bibliography{thesis}

\begin{thebibliography}{10}

\bibitem{bs03}
S.~Baumer and R.~Schuler.
\newblock Improving a probabilistic 3-{SAT} algorithm by dynamic search and
  independent clause pairs.
\newblock In {\em Theory and Applications of Satisfiability Testing}, volume
  2919 of {\em Lecture Notes in Computer Science}, pages 150--161. Springer
  Berlin / Heidelberg, 2004.

\bibitem{thesis}
T.~Hertli.
\newblock Investigating and improving the {PPSZ} algorithm for {SAT}, master's
  thesis.
\newblock ETH Z{\"u}rich, 2010.
\newblock doi: \verb|http://dx.doi.org/10.3929/ethz-a-006206989|.

\bibitem{hofmeister02}
T.~Hofmeister, U.~Sch{\"o}ning, R.~Schuler, and O.~Watanabe.
\newblock A probabilistic 3-{SAT} algorithm further improved.
\newblock In {\em S{TACS} 2002}, volume 2285 of {\em Lecture Notes in Comput.
  Sci.}, pages 192--202. Springer, Berlin, 2002.

\bibitem{istt10}
K.~Iwama, K.~Seto, T.~Takai, and S.~Tamaki.
\newblock Improved randomized algorithms for 3-{SAT}.
\newblock In {\em Algorithms and Computation}, volume 6506 of {\em Lecture
  Notes in Computer Science}, pages 73--84. Springer Berlin / Heidelberg, 2010.

\bibitem{it04}
K.~Iwama and S.~Tamaki.
\newblock Improved upper bounds for 3-{SAT}.
\newblock In {\em Proceedings of the {F}ifteenth {A}nnual {ACM}-{SIAM}
  {S}ymposium on {D}iscrete {A}lgorithms}, pages 328--329 (electronic), New
  York, 2004. ACM.

\bibitem{ppszconference}
R.~Paturi, P.~Pudl{\'a}k, M.~E. Saks, and F.~Zane.
\newblock {An Improved Exponential-Time Algorithm for k-SAT}.
\newblock In {\em Proceedings of the 39th Annual Symposium on Foundations of
  Computer Science}, pages 628--637. IEEE Computer Society, 1998.

\bibitem{ppsz}
R.~Paturi, P.~Pudl{\'a}k, M.~E. Saks, and F.~Zane.
\newblock An improved exponential-time algorithm for {$k$}-{SAT}.
\newblock {\em J. ACM}, 52(3):337--364 (electronic), 2005.

\bibitem{ppz}
R.~Paturi, P.~Pudl{\'a}k, and F.~Zane.
\newblock Satisfiability coding lemma.
\newblock {\em Chicago J. Theoret. Comput. Sci.}, pages Article 11, 19 pp.
  (electronic), 1999.

\bibitem{rolf2006}
D.~Rolf.
\newblock {Improved Bound for the PPSZ/Sch{\"o}ning-Algorithm for 3-SAT}.
\newblock {\em Journal on Satisfiability, Boolean Modeling and Computation},
  1:111--122, 2006.

\bibitem{schoning1999}
U.~Sch{\"o}ning.
\newblock A probabilistic algorithm for {$k$}-{SAT} and constraint satisfaction
  problems.
\newblock In {\em 40th {A}nnual {S}ymposium on {F}oundations of {C}omputer
  {S}cience ({N}ew {Y}ork, 1999)}, pages 410--414. IEEE Computer Soc., Los
  Alamitos, CA, 1999.

\bibitem{welzl05}
E.~Welzl.
\newblock Boolean satisfiability -- combinatorics and algorithms (lecture
  notes), 2005.
\newblock \verb|http://www.inf.ethz.ch/~emo/SmallPieces/SAT.ps|.

\end{thebibliography}

\end{document}